\documentclass[11pt]{article}
\usepackage{authblk}
\usepackage[margin=0.8in]{geometry}
\usepackage{amsmath,amssymb,amsthm}
\usepackage{graphicx}
\usepackage{booktabs}
\usepackage{microtype}
\ifdefined\IOPSTYLE
\else
  \usepackage{xcolor}
  \usepackage[colorlinks=true,linkcolor=blue!60!black,citecolor=blue!60!black,urlcolor=blue!60!black]{hyperref}
\fi

\graphicspath{{figures/}}

\newtheorem{theorem}{Theorem}
\newtheorem{lemma}{Lemma}
\newtheorem{corollary}{Corollary}

\newtheorem{definition}{Definition}

\newtheorem{example}{Example}

\newcommand{\feas}{\mathcal{F}}
\newcommand{\Z}{\mathbb{Z}}
\newcommand{\R}{\mathbb{R}}

\newcommand{\ket}[1]{|#1\rangle}
\newcommand{\bra}[1]{\langle #1|}
\newcommand{\braket}[2]{\langle #1|#2\rangle}
\newcommand{\vmax}{v_{\max}}

\begin{document}

  \title{{A $\lambda$-Resonance Theorem for Constrained QAOA:
Spectral Structure and Recurrence Bounds}
\author{Krit Grover \\ \texttt{krit.grover@mail.utoronto.ca}}
  \affil{Department of Computer and Mathematical Sciences, University of Toronto Scarborough\\
  Toronto, Canada}}

  \maketitle

\begin{abstract}
Training the variational angles of the Quantum Approximate Optimization Algorithm once
on a small instance and reusing them on larger ones, known as \textit{parameter transfer}, is the
standard route past the exact-simulation wall. Existing explanations of parameter transfer focus
primarily on structural similarity, but constrained encodings introduce a second
quantity that must also be tracked: the penalty scale.
{
We characterize the dependence on the deployment penalty scale when
the circuit angles are held fixed.} We show that for a
scalarized cost $C_\lambda=C+\lambda v$ with integer-valued violation count
$v$, the feasible probability $F(\lambda)$ at arbitrary fixed depth $p$ is a
finite real trigonometric polynomial whose frequencies are
  {integer linear combinations of the fixed $\gamma$ angles}. We call this dependence
  $\lambda$-resonance.  It recasts penalty mismatch as a phase-interference
  problem and yields an explicit spectral-support formula, a Bernstein lower
  bound on the detuning required to change $F$ by a specified amount, and exact
  recurrences when the layer phases share a common period.
  {These statements hold for arbitrary fixed angles, including untrained
  angles; they establish spectral support and recurrence conditions without
  asserting that the training penalty is a preferred operating point.}
  Exact-statevector
  experiments on a 20-qubit resource-allocation QUBO exhibit sharp and blunt
  response profiles and empirical recurrences consistent with the predicted
  {spectral support}. {
  Independent penalty-phase sampling recovers the full Fourier coefficients
  and allows coefficient-based recurrence bounds to be evaluated.} A trained
two-qubit example shows that preserving the training phase scale need not
maximize feasibility: a nearby penalty can perform better. Together, the theory and experiments motivate a
conditional transfer protocol: treat $(\beta,\gamma,\lambda_0)$ as a single
calibrated parameter set, verify target validity, and bracket the matched
penalty with a small fixed-angle sweep when feasible.
\end{abstract}

\section{Introduction}
\label{sec:intro}

The Quantum Approximate Optimization Algorithm
(QAOA)~\cite{farhi2014qaoa,zhou2020qaoa} prepares a variational state by
alternating a problem-dependent phase operator with a transverse-field mixer,
and its practical viability on near-term hardware rests on avoiding per-instance
retraining of the angles $(\beta,\gamma)$. A large body of work has shown that
objective values at fixed angles can concentrate across typical
instances~\cite{brandao2018concentration}, that fixed angles perform well across entire graph
families~\cite{wurtz2021fixedangle}, and that angles trained on small graphs
transfer to much larger ones when the instances share local structure: subgraph composition and degree
  parity~\cite{galda2021transferability}, learned graph
  embeddings~\cite{falla2024grl}, higher-locality hypergraph
  reweighting rules~\cite{ornl2026hypergraph}, or edge-weight distributions that can
  be rescaled~\cite{shaydulin2023weighted}.  In
  particular, multiplying an entire cost Hamiltonian by $a$ is exactly
  compensated at the circuit level by $\gamma_k\mapsto\gamma_k/a$~%
  \cite{shaydulin2023weighted}.  In short: the transfer literature is a theory of \textit{structure}.

A separate literature treats constrained problems. Since hard constraints
cannot be expressed natively in an unconstrained binary quadratic model, the
standard encoding adds a penalty term $\lambda v(x)$, where $v(x)$ counts
constraint violations and $\lambda>0$ is a weight to be
chosen~\cite{lucas2014ising}. Choosing $\lambda$ is understood as an
\emph{energetic} trade-off: too small and infeasible states remain attractive,
whereas an unnecessarily large penalty can dominate the objective scale.
Accordingly, the literature derives static or sequential weights from problem
bounds~\cite{ayodele2022penalty,garcia2022exact}, develops scalable
pre-computation procedures~\cite{scalablepen2026}, promotes layer penalties to
variational schedules~\cite{feasdriven2026}, or replaces the
quadratic penalty altogether~\cite{montanez2024unbalanced}. In short: the
penalty literature is a theory of \emph{energetics}.

This paper studies the intersection of structural transfer, penalty
energetics, and phase interference.  Our starting observation is elementary:
in the QAOA circuit the penalty weight never appears
except through the phases $e^{-i\gamma_k\lambda v(x)}$ accumulated at each
layer.  Changing $\lambda$ therefore changes only the
penalty phases, unlike the whole-Hamiltonian rescaling considered in
weighted-cost parameter transfer~\cite{shaydulin2023weighted}. Building on
known integer-spectrum periodicity and general circuit-Fourier
representations~\cite{zhou2020qaoa,sureshbabu2024parameter,schuld2021fourier},
our contribution is the feasible-projector
specialization, including its explicit path-sum coefficients and bounded
frequency support, together with the resulting Bernstein detuning and
partial-rephasing bounds. Section~\ref{sec:theorem} formalizes this pullback and distinguishes the
consequences inherited from known periodicity from the quantitative bounds
derived here.

\paragraph{Contributions.}
\begin{enumerate}
\item \textbf{A $\lambda$-Resonance Theorem}
  (Section~\ref{sec:theorem}): for any
  integer-valued penalty and \emph{arbitrary} fixed angles, the feasible
  probability mass $F(\lambda)$ is a finite real trigonometric polynomial whose
  frequencies are {integer linear combinations of the fixed $\gamma$ angles},
  with an explicit path-sum coefficient formula.
\item \textbf{Robustness and Recurrence Statements}: Bernstein's
  inequality lower-bounds the detuning required for a specified drop, while a
  common torus period produces an exact recurrence.  We state explicitly when
  the partial-recurrence bound is informative and when it is vacuous.
\item \textbf{An Exact-Statevector Study}
  (Section~\ref{sec:experiments}): on a 20-qubit constrained allocation
  instance we report the matched-point behavior, empirical widths, recurrence
  features, and {a direct depth-one control, exact torus
  coefficients, and coefficient-based recurrence bounds}.
\end{enumerate}

We emphasize that the theorem is a property of the QAOA ansatz itself: it
depends only on the fixed angles, not on how they were obtained.

\section{Setting}
\label{sec:setting}

\subsection{Notation and basic definitions}
\label{sec:notation}

Throughout, $n$ is the number of qubits and $N=2^n$.

\begin{itemize}
\item \textbf{Bitstrings and the computational basis.} A bitstring is
  $x=(x_1,\dots,x_n)\in\{0,1\}^n$. The Hilbert space of $n$ qubits has the
  orthonormal \emph{computational basis} $\{\ket{x}\}_{x\in\{0,1\}^n}$; a pure
  state is $\ket{\psi}=\sum_x \psi(x)\ket{x}$ with complex amplitudes
  $\psi(x)=\braket{x}{\psi}$, $\sum_x|\psi(x)|^2=1$. Measuring in this basis
  returns $x$ with probability $|\psi(x)|^2$; a QAOA run is nothing but
  repeated preparation and measurement, so the physical output is the
  probability distribution $|\psi(\cdot)|^2$.
\item \textbf{Diagonal operators.} For $f:\{0,1\}^n\to\R$, write
  $\mathrm{diag}(f)=\sum_x f(x)\ket{x}\!\bra{x}$. Its exponential acts by pure
  phases,
  $e^{-i\gamma\,\mathrm{diag}(f)}\ket{x}=e^{-i\gamma f(x)}\ket{x}$:
  a diagonal evolution changes no measurement probability by itself; it only
  \emph{imprints phases} that later operations can convert into interference.
\item \textbf{Subset probability.} For $S\subseteq\{0,1\}^n$ the projector is
  $\Pi_S=\sum_{x\in S}\ket{x}\!\bra{x}$, and
  $\bra{\psi}\Pi_S\ket{\psi}=\sum_{x\in S}|\psi(x)|^2$ is the probability that
  a measurement lands in $S$.
\end{itemize}

\subsection{Constrained problems and the penalty method}

\begin{definition}[Violation count and feasible set]
\label{def:violation}
A \emph{violation count} is a function
$v:\{0,1\}^n\to\{0,1,\dots,\vmax\}$ taking non-negative \emph{integer} values,
with $\vmax=\max_x v(x)$. The \emph{feasible set} is
$\feas=\{x: v(x)=0\}$, assumed non-empty. Integrality is the only structural
assumption this paper needs, and it is automatic whenever $v$ literally counts
violated constraints.
\end{definition}

\begin{definition}[Penalized scalarization]
\label{def:penalty}
Given a bounded objective $C:\{0,1\}^n\to\R$ and
a penalty weight $\lambda>0$, the \emph{penalized cost} is
\begin{equation}
C_\lambda(x) \;=\; C(x) \;+\; \lambda\, v(x).
\label{eq:cost}
\end{equation}
\end{definition}

The penalty method exists because unconstrained binary samplers: quantum annealers, QAOA, and QUBO heuristics alike
operate on all of $\{0,1\}^n$: an illegal assignment remains in the state
space and is made expensive rather than excluded. Constraint-preserving
mixers are a different design choice, discussed below.  The weight $\lambda$ is the exchange rate between
objective value and legality: a state with $v(x)=2$ carries a surcharge of
$2\lambda$ over its objective value.

The canonical violation count, and the one in our experiments, is the pairwise
\emph{at-most-one} penalty. If bits are partitioned into groups
$G_1,\dots,G_r\subset\{1,\dots,n\}$ of which each may contain at most one
active bit, then
\begin{equation}
v(x)\;=\;\sum_{j=1}^{r}\ \sum_{\substack{a<b\\ a,b\in G_j}} x_a x_b ,
\qquad \vmax=\sum_{j=1}^{r} \binom{|G_j|}{2},
\label{eq:onehot}
\end{equation}
which counts exactly the violated exclusivity pairs: each product $x_ax_b$
equals $1$ precisely when both bits of a forbidden pair are set. The smallest
instance is a single group $G=\{1,2\}$: $v(x)=x_1x_2$ penalizes only the state
$11$, and $\vmax=1$. In our 20-qubit testbed, $r=6$ groups of size $2$ give
$\vmax=6$; the 28-qubit sibling has $r=10$ and $\vmax=10$. More generally,
$\vmax$ grows with both the number and size of the groups.  This enlarges the
allowed bandwidth in Corollary~\ref{cor:vmax}, but it does not force an actual
response curve to narrow.

\subsection{QAOA at fixed angles}

Depth-$p$ QAOA alternates two operations, each carrying one real parameter per
layer: a \emph{phase separator}, which is the exponential of the diagonal cost
operator, and a \emph{mixer}
\begin{equation}
M(\beta)\;=\;e^{-i\beta\sum_{j=1}^{n} X_j}\;=\;\prod_{j=1}^{n} e^{-i\beta X_j},
\end{equation}
a product of identical single-qubit rotations. The mixer
is the only non-diagonal element of the circuit, and its matrix elements have
a closed form we will need:

\begin{lemma}[Mixer matrix elements]
\label{lem:mixer}
For all $x,y\in\{0,1\}^n$,
\begin{equation}
\bra{y}M(\beta)\ket{x}
\;=\;(\cos\beta)^{\,n-d(x,y)}\,(-i\sin\beta)^{\,d(x,y)}
\end{equation}
where $d(x,y)$ is the Hamming distance between $x$ and $y$.
\end{lemma}

\begin{proof}
On one qubit, $$e^{-i\beta X}=\cos(\beta) I - i\sin(\beta) X$$ by trivial Taylor expansion. So
$\bra{y_j}e^{-i\beta X}\ket{x_j}$ equals $\cos\beta$ if $y_j=x_j$ and
$-i\sin\beta$ if $y_j\neq x_j$. The claim follows by taking the tensor
product over the $n$ coordinates: agreeing coordinates contribute
$\cos\beta$, disagreeing ones $-i\sin\beta$, and there are $d(x,y)$ of the
latter.
\end{proof}
Two facts in this formula do real work later. First, when
$\sin\beta\cos\beta\neq0$ the matrix element is nonzero for \emph{every} pair
$(x,y)$: the mixer connects every bitstring to every other, in particular
feasible states to infeasible ones. Second, only the mixer moves amplitude
between bitstrings; the phase separator is diagonal. Interference between
paths of different violation count
therefore requires nontrivial mixing.  A layer with
$\beta_k\in(\pi/2)\mathbb Z$ ceases to branch computational-basis paths, but
other layers can still generate interference. With angles
$\beta=(\beta_1,\dots,\beta_p)$,
$\gamma=(\gamma_1,\dots,\gamma_p)$, the QAOA state for the penalized cost
\eqref{eq:cost} is
\begin{equation}
\ket{\psi(\lambda)} \;=\;
\prod_{k=1}^{p}\Big[\, M(\beta_k)\; e^{-i\gamma_k \,\mathrm{diag}(C_\lambda)}\Big]\,
\ket{+}^{\otimes n},
\label{eq:qaoa}
\end{equation}
read right-to-left: uniform superposition, then alternately imprint cost
phases and mix, $p$ times. Note that adding any constant to $C_\lambda$
multiplies the state by a global phase and changes no probability.
{Multiplicative normalization changes the
fixed-angle circuit and must therefore be specified together with
the objective coefficients. The conventions used in our experiments
are given in Section~\ref{sec:testbed}.}

\begin{definition}[Feasible fraction]
\label{def:F}
The object of study is the probability that a measurement of
$\ket{\psi(\lambda)}$ returns a feasible bitstring,
\begin{equation}
F(\lambda)\;=\;\bra{\psi(\lambda)}\,\Pi_\feas\,\ket{\psi(\lambda)}
\;=\;\sum_{x\in\feas}\big|\braket{x}{\psi(\lambda)}\big|^{2}.
\label{eq:F}
\end{equation}
\end{definition}

$F$ is the headline resource of penalty-based sampling in workflows that
discard infeasible shots, where it multiplies the effective shot budget.

Throughout, the angles are \emph{fixed}: we study $F$ as a function of the
deployment penalty weight $\lambda$ alone, with the objective, constraints, and
angles held constant.  This isolates the
penalty-detuning component of transfer rather than reproducing every change in
a cross-instance deployment.  Section~\ref{sec:theorem}
establishes the finite spectral structure, while the case study tests whether
that structure has a visible manifestation for the reported angles.

\section{The $\lambda$-Resonance Theorem}
\label{sec:theorem}

The plan of this section: we first recall the two pieces of classical analysis
the argument rests on (trigonometric polynomials and Bernstein's inequality),
then expand the QAOA amplitude as a sum over computational-basis paths and
observe that $\lambda$ enters each path only through a phase. The theorem,
its coefficient structure, and the
corollaries all follow from that single observation.

\subsection{Analytic preliminaries}

\begin{definition}[Finite trigonometric polynomial]
\label{def:trig}
A \emph{finite trigonometric polynomial} with frequency set
$\Omega\subset\R$, $|\Omega|<\infty$, is a function
$$T(\lambda)=\sum_{\omega\in\Omega}A_\omega e^{-i\omega\lambda}$$ with complex
coefficients $A_\omega$. $T$ is real-valued on $\R$ iff $\Omega=-\Omega$ and
$A_{-\omega}=A_\omega^{*}$ when $\Omega$ is taken to be the actual support
(zero-coefficient frequencies removed). Writing $\Omega_{\max}=\max_{\omega\in\Omega}
|\omega|$, the function $T$ extends to an entire function of the complex
variable $\lambda$ of \emph{exponential type} $\Omega_{\max}$ (i.e.\
$|T(\lambda)|\le \mathrm{const}\cdot e^{\Omega_{\max}|\mathrm{Im}\,\lambda|}$):
frequencies bound growth off the real axis, and, by the next lemma, slopes on
it.
\end{definition}

\begin{lemma}[Bernstein's inequality~\cite{boas1954}]
\label{lem:bernstein}
If $T$ is entire of exponential type $\sigma$ and bounded on $\R$, then
$\sup_\R|T'|\le \sigma\,\sup_\R|T|$. In particular, a finite real
trigonometric polynomial obeys
$\|T'\|_\infty\le\Omega_{\max}\|T\|_\infty$, and iterating,
$\|T''\|_\infty\le\Omega_{\max}^2\|T\|_\infty$.
\end{lemma}

Bernstein's inequality converts the bounded frequency support into a bound on
how rapidly $F$ can vary with $\lambda$, yielding a lower bound on the
detuning required for a specified change in feasibility.

\subsection{Path expansion of the amplitude}

\begin{lemma}[Path expansion and phase factorization]
\label{lem:path}
For every $y\in\{0,1\}^n$ the QAOA amplitude \eqref{eq:qaoa} can be written
\begin{equation}
\braket{y}{\psi(\lambda)}
\;=\;\sum_{\mathbf{v}\in\{0,\dots,\vmax\}^p}
e^{-i\lambda\,\gamma\cdot\mathbf{v}}\; B_{\mathbf{v}}(y),
\label{eq:pathgroup}
\end{equation}
where $\gamma\cdot\mathbf v=\sum_k\gamma_k v_k$ and the \emph{partial
amplitudes}
\begin{equation}
B_{\mathbf{v}}(y)\;=\;
\frac{1}{\sqrt N}
\sum_{\substack{x_1,\dots,x_p\\ v(x_k)=v_k\ \forall k}}
\ \prod_{k=1}^{p}
\bra{x_{k+1}}M(\beta_k)\ket{x_k}\; e^{-i\gamma_k C(x_k)}
\qquad (x_{p+1}\equiv y)
\label{eq:Bv}
\end{equation}
do not depend on $\lambda$.
\end{lemma}

\begin{proof}
Insert a complete computational basis $\sum_{x_k}\ket{x_k}\!\bra{x_k}$ between
every pair of adjacent operators in \eqref{eq:qaoa}. Since the phase
separators are diagonal, each insertion collapses to a single phase factor,
leaving a sum over \emph{paths} $(x_1,\dots,x_p)$:
\begin{equation}
\braket{y}{\psi(\lambda)}
=\sum_{x_1,\dots,x_p}
\bra{y}M(\beta_p)\ket{x_p}\, e^{-i\gamma_p C_\lambda(x_p)}\cdots
\bra{x_2}M(\beta_1)\ket{x_1}\, e^{-i\gamma_1 C_\lambda(x_1)}\,
\frac{1}{\sqrt N},
\end{equation}
where the trailing factor is $\braket{x_1}{+^{\otimes n}}=N^{-1/2}$ and the
mixer matrix elements are given explicitly by Lemma~\ref{lem:mixer}. Now
split each cost phase using \eqref{eq:cost}:
$e^{-i\gamma_k C_\lambda(x_k)}
= e^{-i\gamma_k C(x_k)}\cdot e^{-i\gamma_k\lambda\, v(x_k)}$.
The first factor is $\lambda$-free; the second depends on the path only
through the integer $v(x_k)$. The total $\lambda$-dependence of a path is
therefore the single scalar phase
$e^{-i\lambda\sum_k\gamma_k v(x_k)}$. Grouping paths by their
\emph{violation profile} $\mathbf v=(v(x_1),\dots,v(x_p))$ yields
\eqref{eq:pathgroup} with $B_{\mathbf v}(y)$ collecting everything
$\lambda$-free, which is \eqref{eq:Bv}.
\end{proof}

The physical reading of Lemma~\ref{lem:path}: the state is a coherent sum of
finitely many ``violation-history channels,'' one per profile $\mathbf v$; the
penalty weight rotates each channel's phase at rate $\gamma\cdot\mathbf v$, and
nothing else. Whether the channels add constructively or destructively on the
feasible subspace is then a pure interference question, which is what we answer in the following theorem.

\subsection{The Theorem}

\begin{theorem}[$\lambda$-Resonance Theorem]
\label{thm:main}
Fix any depth $p$, any angles $(\beta,\gamma)\in\R^{2p}$, any bounded objective
$C$, and any integer-valued penalty $v$ with maximum value $\vmax$. Then the
feasible fraction \eqref{eq:F} is a finite real trigonometric polynomial in
$\lambda$, and its angular frequencies $\omega$ are {integer linear combinations of the fixed $\gamma$ parameters with bounded integer coefficients}:
\begin{equation}
F(\lambda) \;=\; \sum_{\omega\in\Omega} A_\omega\, e^{-i\omega\lambda},
\qquad
\Omega=\Big\{\textstyle\sum_{k=1}^{p}\gamma_k m_k \;:\; m_k\in\Z,\ |m_k|\le \vmax\Big\},
\label{eq:thm}
\end{equation}
with $A_{-\omega}=A_\omega^{*}$ and $F(\lambda)\in[0,1]$ for all $\lambda$.
The coefficients
\begin{equation}
A_\omega\;=\;\sum_{\substack{\mathbf v,\mathbf v'\\
\gamma\cdot(\mathbf v-\mathbf v')=\omega}}\ \sum_{y\in\feas}
B_{\mathbf v}(y)\,B_{\mathbf v'}(y)^{*}
\label{eq:Aomega}
\end{equation}
depend on $(\beta,\gamma)$, on $C$, and on the geometry of $v$, but \emph{not}
on $\lambda$.
\end{theorem}

\begin{proof}
Substitute \eqref{eq:pathgroup} into \eqref{eq:F}:
\begin{equation}
F(\lambda)=\sum_{y\in\feas}\big|\braket{y}{\psi(\lambda)}\big|^2
=\sum_{\mathbf{v},\mathbf{v}'} e^{-i\lambda\,\gamma\cdot(\mathbf{v}-\mathbf{v}')}
\sum_{y\in\feas} B_{\mathbf{v}}(y)\,B_{\mathbf{v}'}(y)^{*}.
\end{equation}
Setting $m_k=v_k-v'_k\in\{-\vmax,\dots,\vmax\}$ gives \eqref{eq:thm} with
coefficients \eqref{eq:Aomega}. Exchanging $\mathbf v\leftrightarrow\mathbf v'$
conjugates the inner sum and negates the frequency, giving
$A_{-\omega}=A_\omega^*$; the range $[0,1]$ is immediate since $F$ is a
probability by Definition~\ref{def:F}.
\end{proof}

Note what the theorem does and does not use: integrality of $v$, finiteness of $\vmax$, and
nothing else. Not the form of $C$, not the training procedure, not even that
the angles are good. The size of the frequency set is at most
$(2\vmax+1)^p$, independent of $n$; the exponential dimension of the problem
is entirely absorbed into the coefficients.

\subsection{Relation to Known Periodicity and Transfer Interpretation}

For an integer-spectrum Hamiltonian, $2\pi$ periodicity in each QAOA phase
angle is known~\cite{zhou2020qaoa,sureshbabu2024parameter}. More generally,
expectation values of repeated parameter-encoding circuits have Fourier
support determined by eigenvalue differences~\cite{schuld2021fourier}.
Theorem~\ref{thm:main} specializes this structure to the one-dimensional
penalty sweep $\alpha_k=\lambda\gamma_k$ and to the feasible-projector
observable $\Pi_\feas$.

The transfer-specific distinction is that changing $\lambda$ modifies only
the penalty phase while leaving the objective phase fixed. Under a
whole-Hamiltonian rescaling $H\mapsto aH$, the original circuit is recovered
by $\gamma_k\mapsto\gamma_k/a$, a relation used by Shaydulin \emph{et
al.} in weighted-MaxCut parameter transfer~\cite{shaydulin2023weighted}.
By contrast, for $H(\lambda)=C+\lambda v$, reproducing both phase factors with
a common rescaled angle $\gamma'_k$ would generically require
\[
\gamma'_k C=\gamma_k C,
\qquad
\gamma'_k\lambda'v=\gamma_k\lambda v.
\]
When $C$ and $v$ are both nontrivial, these conditions generically imply
$\gamma'_k=\gamma_k$ and $\lambda'=\lambda$. Operationally, $\lambda$ is
therefore not redundant with $\gamma$: for fixed transferred angles, it
selects an operating point along the penalty-phase trajectory. Matching the
donor value preserves the products $\gamma_k\lambda$ under a common
normalization of $v$, but it does not by itself maximize the feasible
probability.

The corollaries below have different logical status. If
$\gamma_k\Delta\in2\pi\mathbb{Z}$ for every layer, exact recurrence follows
directly from the known coordinatewise periodicity. The detuning and
partial-rephasing bounds require the additional facts that the Fourier degree
is finite and that $0\leq F\leq1$; Bernstein's inequality then limits how
rapidly feasibility can change with $\lambda$. These bounds constrain worst-case sensitivity and recurrence error, but they
do not locate peaks, predict empirical peak widths, or imply that the observed
response must narrow as $\vmax$ increases.

\subsection{The torus picture}

\begin{lemma}[Torus representation]
\label{lem:torus}

Introduce independent penalty-phase angles $\alpha=(\alpha_1,\ldots,\alpha_p)$
and the physical response
\begin{equation}
\begin{split}
\ket{\psi(\alpha)}={}&
\prod_{k=1}^{p}\left[M(\beta_k)e^{-i\gamma_k\,\mathrm{diag}(C)}
e^{-i\alpha_k\,\mathrm{diag}(v)}\right]\ket{+}^{\otimes n},\\
G(\alpha)={}&\bra{\psi(\alpha)}\Pi_\feas\ket{\psi(\alpha)}.
\end{split}
\label{eq:Gphysical}
\end{equation}

Then $G:\mathbb{T}^p\to[0,1]$ is $2\pi$-periodic in every coordinate and is a
finite Fourier series with integer degree at most $\vmax$ in each coordinate.
Moreover,
\begin{equation}
F(\lambda)\;=\;G\big(\gamma_1\lambda \bmod 2\pi,\;\dots,\;\gamma_p\lambda \bmod 2\pi\big).
\end{equation}
\end{lemma}

\begin{proof}

Equation~\eqref{eq:Gphysical} makes $G\in[0,1]$ immediate.  Repeating the path
expansion of Lemma~\ref{lem:path} with $\alpha_k$ independent gives phases
$e^{-i\alpha\cdot(\mathbf v-\mathbf v')}$, hence integer Fourier indices
$m_k=v_k-v'_k$ with $|m_k|\le\vmax$.  Integrality of $v$ gives coordinatewise
$2\pi$ periodicity.  Finally, the commuting diagonal factors obey
$e^{-i\gamma_k\mathrm{diag}(C_\lambda)}=
e^{-i\gamma_k\mathrm{diag}(C)}e^{-i\lambda\gamma_k\mathrm{diag}(v)}$;
setting $\alpha_k=\lambda\gamma_k$ recovers \eqref{eq:qaoa} and $F$.

\end{proof}

Thus, for fixed $\gamma$, varying $\lambda$ moves the circuit along the
trajectory
$\theta(\lambda)=\lambda\gamma\bmod 2\pi$
on the penalty-phase torus.  The vector $\gamma$ determines the direction and
rate of this motion, while $\lambda$ selects the operating point along the
trajectory. All theorem-level
$\lambda$-dependence in this paper is a statement about that line.

{
\paragraph{Exact recovery of the torus coefficients.}
Write $\mathcal M=\{-\vmax,\ldots,\vmax\}^p$ and
\begin{equation}
G(\alpha)=\sum_{\mathbf m\in\mathcal M} A_{\mathbf m}
e^{-i\mathbf m\cdot\alpha},\qquad
A_{-\mathbf m}=A_{\mathbf m}^{*}.
\label{eq:toruscoeff}
\end{equation}
For $q=2\vmax+1$, discrete Fourier orthogonality gives
\begin{equation}
A_{\mathbf m}=\frac{1}{q^p}
\sum_{j_1,\ldots,j_p=0}^{q-1}
G\!\left(\frac{2\pi}{q}\mathbf j\right)
e^{2\pi i\mathbf m\cdot\mathbf j/q}.
\label{eq:exactdft}
\end{equation}
Indeed, the difference between any two supported indices lies in
$[-2\vmax,2\vmax]$ in each coordinate; the only multiple of $q$ in
this range is zero. Thus this grid recovers every coefficient without
aliasing. The objective phases $\gamma_k C$ remain fixed while the
penalty phases $\alpha_k$ vary independently. For the $p=2$, $\vmax=6$
testbed, the grid is $13\times13$, requiring 169 circuit evaluations per
scalarization. This is an exact finite-series identity; numerical
statevector evaluations introduce floating-point error, which can be
checked by reconstructing independently evaluated responses.
The one-dimensional frequency is $\omega_{\mathbf m}=\gamma\cdot\mathbf m$.
If several indices yield the same frequency, their complex coefficients
sum to the coefficient $A_\omega$ in Theorem~\ref{thm:main}.
}

\paragraph{Use of the word ``resonance''.}
Here \emph{resonance} denotes the structured, oscillatory dependence of
$F(\lambda)$ on the penalty scale induced by phase interference. Theorem
\ref{thm:main} characterizes the allowed Fourier support, while the
corollaries below bound its sensitivity and give conditions for recurrence.
None of these results requires the training penalty $\lambda_0$, or a
single-layer rephasing point, to be a peak. Peak placement and empirical width
depend on the realized Fourier coefficients and must be determined for each
angle set.


\begin{example}[Energy optimization need not maximize feasibility]
\label{ex:flank}
Let $p=1$, $v(x)=x_1x_2$, $\lambda_0=2$, and, in the basis
$(00,01,10,11)$, let $C=(0,1,2,-1)$.  The training Hamiltonian is therefore
$C_{\lambda_0}=(0,1,2,1)$ and has the unique feasible ground state $00$.
For the expected-energy objective, one global minimum follows from the closed form
\begin{equation}
E(\beta,\gamma)=1+\frac{\sin(2\beta)\sin(2\gamma)
+\sin(4\beta)\sin\gamma(1+\cos\gamma)}{4}
\end{equation}
as follows.  Set $x=\cos\gamma$ and $y=\cos(2\beta)$.  Minimizing $E$ is
equivalent to maximizing
\begin{equation}
Q(x,y)=\sqrt{(1-x^2)(1-y^2)}\,|x+(1+x)y|
\end{equation}
on $[-1,1]^2$.  The boundary gives $Q=0$; the nonzero interior stationary
equations reduce to $x=y$ and
$(x^2+x-1)(2x+1)=0$.  The root $x=-1/2$ gives $Q=9/16$, whereas
$x=y=(\sqrt5-1)/2$ gives the global maximum $Q=1$.  Choosing the signs that
lower $E$ gives the trained optimum
\begin{equation}
\gamma=\arccos\!\left(\frac{\sqrt5-1}{2}\right)=0.904557\ldots,
\qquad \beta=-\gamma/2=-0.452278\ldots,
\end{equation}
with $E=1/2$.  At these fixed angles,
$F(\lambda_0)=(5-\sqrt5)/4=0.690983\ldots$ but
$F'(\lambda_0)=0.177780\ldots\ne0$; the matched penalty is on a flank.
For reference, the same curve reaches $F\simeq0.9794$ near
$\lambda=4.151$.  Thus, under a common normalization of $v$, matching preserves
the numerical penalty-phase products $\gamma_k\lambda_0$ but provides no
extremum guarantee.  A practitioner should validate or improve the match with
a gate-valid local sweep, optimize over the permitted deployment range, or
train a robust objective over a $\lambda$ interval.
\end{example}

\subsection{Corollaries}
\label{sec:corollaries}

\begin{corollary}[Detuning lower bound]
\label{cor:width}
Let $\Omega_{\max}=\vmax\sum_k|\gamma_k|$. Then, for all $\lambda$ and $\Delta$,
\begin{equation}
\big|F(\lambda+\Delta)-F(\lambda)\big|
\leq \frac{\Omega_{\max}}{2}|\Delta|.
\end{equation}
Consequently, if $\Omega_{\max}>0$ and
$\big|F(\lambda+\Delta)-F(\lambda)\big|\geq\delta>0$, then
\begin{equation}
|\Delta|
\geq \frac{2\delta}{\Omega_{\max}}
=\frac{2\delta}{\vmax\sum_k|\gamma_k|}.
\label{eq:widthbound}
\end{equation}
If $\Omega_{\max}=0$, then $F$ is constant and no positive change occurs.
This is a worst-case lower bound on the detuning required for a specified
change in feasibility. It is not a prediction of a peak width.

\end{corollary}

\begin{proof}
Every frequency in \eqref{eq:thm} satisfies
\[
|\omega|
=\left|\sum_k\gamma_km_k\right|
\leq\vmax\sum_k|\gamma_k|
=\Omega_{\max}.
\]
Thus $F$ is a finite real trigonometric polynomial of exponential type at
most $\Omega_{\max}$. Apply Bernstein's inequality to the auxiliary centered
polynomial $T=F-1/2$. Subtracting a constant does not change its derivative
or nonzero-frequency support, while $0\leq F\leq1$ implies
$\|T\|_\infty\leq1/2$. Therefore,
\[
\|F'\|_\infty=\|T'\|_\infty
\leq\Omega_{\max}\|T\|_\infty
\leq\frac{\Omega_{\max}}{2}.
\]

The finite-difference inequality in the corollary follows from the mean value
theorem, and rearranging it gives \eqref{eq:widthbound}.
\end{proof}

\begin{corollary}[Revivals]
\label{cor:revivals}
(i) \emph{Exact revivals:} if $\gamma_k\Delta\in 2\pi\Z$ for every
$k=1, \dots,p$, then $F(\lambda_0+\Delta)=F(\lambda_0)$ for every $\lambda_0$.
(ii) \emph{Partial revivals:} if $\gamma_k\Delta\in2\pi\Z$ for the layers
$k\in S$ only, then

\begin{equation}
\big|F(\lambda_0+\Delta)-F(\lambda_0)\big|\;\le\;
\min\!\left\{1,\;\frac{\vmax}{2} \sum_{k\notin S}
\mathrm{dist}\!\left(\gamma_k\Delta,\,2\pi\Z\right)\right\}.
\label{eq:partialrevival}
\end{equation}
Single-layer rephasing values
$\lambda_0\pm 2\pi j/\gamma_{\mathrm{dom}}$ are therefore only
\emph{candidate} near-recurrence locations.  The statement is informative
only when its nontrivial right-hand side is below $1$; it does not say that
such a location is a local maximum.

{
The realized coefficients provide the additional bound
\begin{equation}
\big|F(\lambda_0+\Delta)-F(\lambda_0)\big|
\le
\sum_{\substack{\mathbf m\in\mathcal M:\\
                  \exists k\notin S\text{ with }m_k\ne0}}
|A_{\mathbf m}|\,\big|e^{-i\omega_{\mathbf m}\Delta}-1\big|.
\label{eq:coefficientrevival}
\end{equation}
This may be combined with \eqref{eq:partialrevival} by taking the smaller
bound. In particular, for $p=2$ and $\Delta=2\pi/\gamma_1$ with
$\gamma_1\ne0$, the sum in \eqref{eq:coefficientrevival} is restricted to
$m_2\ne0$. Whether it is below the trivial bound of 1 depends on the
computed coefficients, not only on their allowed support.
}

\end{corollary}

\begin{proof}
By Lemma~\ref{lem:torus}, $F(\lambda)=G(\theta(\lambda))$ with
$\theta_k(\lambda)=\gamma_k\lambda \bmod 2\pi$. (i) The hypothesis makes
$\theta(\lambda_0+\Delta)=\theta(\lambda_0)$ on the torus. (ii) Write
$\delta_k=\gamma_k\Delta \bmod 2\pi$ taken in $(-\pi,\pi]$, so $\delta_k=0$
for $k\in S$ and $|\delta_k|=\mathrm{dist}(\gamma_k\Delta,2\pi\Z)$ otherwise.
For fixed values of the other coordinates, $\theta_k\mapsto G(\theta)$ is a
univariate trigonometric polynomial of degree at most $\vmax$
(Lemma~\ref{lem:torus}).  Centering it at $1/2$ as above, Bernstein's
inequality gives
$\|\partial G/\partial\theta_k\|_\infty\le\vmax/2$. Moving from
$\theta(\lambda_0)$ to $\theta(\lambda_0)+\delta$ one coordinate at a time and
telescoping gives the second argument of the minimum in
\eqref{eq:partialrevival}; the trivial probability bound supplies the first.
{
For \eqref{eq:coefficientrevival}, subtract the Fourier expansions at the
two penalties:
\[
F(\lambda_0+\Delta)-F(\lambda_0)
=\sum_{\mathbf m\in\mathcal M} A_{\mathbf m}
e^{-i\omega_{\mathbf m}\lambda_0}
\left(e^{-i\omega_{\mathbf m}\Delta}-1\right).
\]
If $m_k=0$ for every $k\notin S$, the last factor vanishes by the
rephasing hypothesis. Applying the triangle inequality to the remaining
terms gives the stated bound.
}
\end{proof}

\begin{corollary}[$\vmax$-dependence of the allowed bandwidth]
\label{cor:vmax}

For fixed $\gamma$ and a fixed normalization of $v$, increasing $\vmax$
enlarges the allowed Fourier bandwidth $\Omega_{\max}$ and weakens the
guaranteed minimum detuning in \eqref{eq:widthbound} in proportion to
$1/\vmax$.
For the at-most-one penalties in \eqref{eq:onehot},
\[
\vmax=\sum_j\binom{|G_j|}{2},
\]
which grows with the number and size of the constraint groups. This permits
sharper $\lambda$-dependence but does not require it, because the coefficients
of the newly allowed high-frequency terms may vanish. In particular,
Corollary~\ref{cor:limits}(i) gives a flat response for every $\vmax$.
Whether empirical peak widths narrow with instance size is therefore a
family-dependent question that requires an otherwise matched, fixed-angle
sweep; it is not determined by the bound.

\end{corollary}

\begin{corollary}[Limits and training-independence]
\label{cor:limits}
(i) At $\gamma=0$, $F\equiv|\feas|/2^n$: the uniform-sampling baseline,
$\lambda$-independent. (ii) Nothing in Theorem~\ref{thm:main} depends on how
$(\beta,\gamma)$ were obtained: the {theorem} applies equally to trained, hand-set,
and random angles.  The angle-generation method alone implies no conclusion
about peak placement.
\end{corollary}

\begin{proof}[Proof of (i)]
With $\gamma=0$ every phase separator in \eqref{eq:qaoa} is the identity, and
since $X\ket{+}=\ket{+}$ implies
$M(\beta)\ket{+}^{\otimes n}=e^{-in\beta}\ket{+}^{\otimes n}$, the prepared
state is $\ket{+}^{\otimes n}$ up to a global phase. Its measurement
distribution is uniform, so $F=|\feas|/2^n$ for every $\lambda$. Claim (ii) is
a property of the statement of Theorem~\ref{thm:main}, whose hypotheses
mention the angles only as fixed real numbers.
\end{proof}

\section{\texorpdfstring{Exact-statevector case study}{Exact-statevector case study}}
\label{sec:experiments}

\subsection{Testbed}
\label{sec:testbed}

The instance is a multi-user resource-allocation QUBO from a semantic-network
scheduling pipeline: $2$ cells $\times$ $2$ users
each $\times$ $3$ shared resource blocks, plus one rate bit and one power bit
per user: $n=20$ qubits, $N=2^{20}$ states, at-most-one penalty
\eqref{eq:onehot} with $r=6$ groups of size $2$, hence $\vmax=6$;
$|\feas|=186{,}624$ (uniform baseline $F=0.178$). The objective $C$ is a
Dirichlet-weighted sum of three least-squares quadratic surrogates (quality,
latency, energy); the calibrated penalty weight is $\lambda_0=2.19597$.
{
We use the minimization convention
\begin{equation}
s_j(x)=o_j+x^{\mathsf T}Q_jx,\qquad
C_c(x)=\sum_{j=1}^{3}c_js_j(x),\qquad
C_{c,\lambda}(x)=C_c(x)+\lambda x^{\mathsf T}Q_{\rm pen}x,
\label{eq:explicitinstance}
\end{equation}
where $c_j\ge0$ and $\sum_jc_j=1$. Each $Q_j$ is stored as an
upper-triangular matrix, so an off-diagonal coefficient is included once.
The fitted targets are the negatives of the training-sample min--max
normalizations of $(q,-L,-E)$, where $q$, $L$, and $E$ denote quality,
latency, and energy. The three fitted surrogates, their offsets, and
$Q_{\rm pen}$ specify the circuit objective independently of the physical
model used to generate the fit.
For zero-based bit indices, the 20-qubit exclusivity groups are
$\{2j,2j+1\}$, $j=0,\ldots,5$; bits 12--15 encode the four rate choices
and bits 16--19 encode the four power choices. The 28-qubit instances
have five resource blocks per cell, groups $\{2j,2j+1\}$,
$j=0,\ldots,9$, rate bits 20--23, and power bits 24--27.
The rate choices are $(0.2,0.8)$ and the power choices are $(-10,0)$ dBm.
The supplementary materials provide the complete $Q_j$, $o_j$, and
$Q_{\rm pen}$ arrays for both sizes and both modality compositions
in JSON and NPZ formats. They also include all scalarization vectors,
full-precision donor angles, input data, instance-generation code,
and numerical results for the transfer experiments. Numerical
validation checks and execution instructions are included.
}
Everything in this 20-qubit fixed-angle subsection is
evaluated with an exact statevector simulator with no shots, sampling noise, or
hardware, so its reported $F(\lambda)$ values are the exact quantity in
\eqref{eq:F}.

Two depth-$2$ angle sets were trained on this instance at
$\lambda_0$ by multistart Nelder--Mead on the exact simulator. Their
parameters, phase budgets, and training losses are summarized in
Table~\ref{tab:trained-angle-sets}.

{
The ``train loss'' is the mean normalized expected penalized cost,
\begin{equation}
\mathcal L(\beta,\gamma)=\frac{1}{|\mathcal C_{\rm train}|}
\sum_{c\in\mathcal C_{\rm train}}
\frac{\langle\psi_c(\lambda_0)|\mathrm{diag}(C_{c,\lambda_0})
|\psi_c(\lambda_0)\rangle-a_c}{b_c-a_c},
\label{eq:trainloss}
\end{equation}
where $a_c=\min_{x\in\{0,1\}^{20}}C_{c,\lambda_0}(x)$ and
$b_c=\max_{x\in\{0,1\}^{20}}C_{c,\lambda_0}(x)$ include infeasible
bitstrings. A constant-cost summand is defined as zero. The training set
contains the uniform weight vector and four fixed
$\mathrm{Dirichlet}(1,1,1)$ draws. The loss in \eqref{eq:trainloss} is a normalized expected energy.
The circuit phases use $C_{c,\lambda_0}$ without this min--max
rescaling.
}

\begin{table}[htbp]
\caption{{Nominal depth}-$2$ angle sets trained on the 20-qubit instance. The
optimization budget is multistarts $\times$ maximum Nelder--Mead iterations.
{The blunt set is approximately depth one
for the feasible probability on the tested grid: the maximum discrepancy
from the specified $p=1$ control over all three scalarizations is $0.0136$.}}
\label{tab:trained-angle-sets}
\centering
\small
\begin{tabular}{lccccc}
\toprule
angle set & budget & $\beta$ & $\gamma$ & $\sum_k|\gamma_k|$ & train loss\\
\midrule
sharp & $8\times150$ & $(1.300,\,-1.736)$ & $(3.512,\ 0.778)$ & $4.290$ & $0.0187$\\
blunt & $2\times60$  & $(-0.030,\,-0.418)$ & $(0.537,\ 0.075)$ & $0.612$ & $0.0209$\\
\bottomrule
\end{tabular}
\end{table}

The two searches returned numerically close training losses, $0.0187$ and
$0.0209$, while their phase budgets differ by a factor of $7.01$. For this
instance, $\vmax=6$, so Corollary~\ref{cor:width} gives
$\Omega_{\max}=25.742$ for the sharp set and $\Omega_{\max}=3.673$ for the
blunt set. Consequently, for the same specified change in feasibility, the
guaranteed minimum-detuning bound is $7.01$ times larger for the blunt set.
Since the searches used different optimizer budgets, and the
$\beta$ angles differ, the losses should be regarded as
numerically close rather than statistically equivalent.

To characterize the realized response curves, we sweep
$\lambda\in[0,5]$ with step $0.05$ for both angle sets, using the uniform
scalarization and two fixed Dirichlet draws.
{In the objective order of \eqref{eq:explicitinstance},
these vectors are
\begin{equation}
\begin{aligned}
c^{(0)}&=(1/3,\;1/3,\;1/3),\\
c^{(1)}&=(0.395461989543,\;0.593018059491,\;0.011519950966),\\
c^{(2)}&=(0.001039758055,\;0.252155601689,\;0.746804640256).
\end{aligned}
\label{eq:sweepweights}
\end{equation}
The displayed decimals are rounded; the supplementary files retain their
full precision.}
{For the sharp set, we additionally evaluate the independent
$13\times13$ penalty-phase grid of \eqref{eq:exactdft} for each of the three
scalarizations, requiring 507 circuit evaluations in total for coefficient
recovery. The resulting spectrum and its use in the recurrence bound are
reported in Section~\ref{sec:experiments}, Fig.~\ref{fig:spectrum},
and Table~\ref{tab:spectral-lines}.}
Within each trace, the instance, objective scalarization, constraint
structure, and QAOA angles are held fixed; only the deployment penalty changes.
The sweep therefore isolates the dependence on $\lambda$ within each trace,
although the contrast between the two angle sets does not isolate the causal
effect of phase budget.

\paragraph{{Depth-one comparison for the blunt set.}}

If $\beta_1=0$ exactly, the depth-$2$ circuit reduces algebraically to
\begin{equation}
M(\beta_2)e^{-i\gamma_2H}M(0)e^{-i\gamma_1H}
=M(\beta_2)e^{-i(\gamma_1+\gamma_2)H},
\end{equation}
{a depth-$1$ circuit whose rounded parameters are}
$(\beta,\gamma)=(-0.418,0.612)$. The trained value $\beta_1=-0.030$ is close to this reduction point, but the
reduction is not exact.
{
We therefore directly evaluated the specified $p=1$ circuit with
$(\beta,\gamma)=(-0.418,0.612)$ on the same instance, the same three
scalarizations, and all 101 penalties in $\{0,0.05,\ldots,5\}$.
The values of $\max_\lambda|F_{\rm blunt}-F_{p=1}|$ are
$0.0130222$, $0.0135745$, and $0.0130861$ for the three scalarizations,
respectively; their overall maximum is $0.0135745$.
For the mean curve shown in Fig.~\ref{fig:curves}, the corresponding maximum
is $0.0129802$. Thus the nominal $p=2$ blunt circuit is effectively depth
one for this observable to within 1.36 percentage points on the tested
grid. 
Using the full-precision parameters gives an overall maximum
discrepancy of $0.0125790$, confirming the same conclusion.
The blunt curve consequently cannot isolate
the effect of a small phase budget from proximity to a depth-reduction limit.
}

We ask three questions:
\begin{itemize}
    \item \textbf{Q1}: Is $F$ high near $\lambda_0$ for these two trained sets,
    without presuming that $\lambda_0$ is a maximizer?
    \item \textbf{Q2}: What are the empirical widths under a stated crossing
    definition, and how do they compare with the rigorous lower bounds?
    \item \textbf{Q3}: {What are the exact Fourier
    coefficients, and how tightly do they bound the observed recurrence?}
\end{itemize}

\subsection{Results}

\paragraph{Q1: Feasibility near the matched penalty.}
On the $[0,5]$ grid, $\lambda=2.20$ is the nearest sampled point to
$\lambda_0=2.195970387\ldots$. At this point, the blunt and sharp sets attain
$F=0.999$ and $F=0.960$, respectively. This is the largest sampled value for
the blunt set. For the sharp set, however, the largest sampled value on the
grid is slightly higher, $F=0.963$ at $\lambda=4.00$, the grid point nearest the candidate single-layer rephasing location $\lambda=3.9849612\ldots$. Thus penalty matching places both angle sets in a high-feasibility region, but it does not generally
identify the largest value of $F$. This is consistent with
Theorem~\ref{thm:main}, which does not determine peak location, and
Example~\ref{ex:flank} shows that even a local extremum is not guaranteed.

\begin{figure}[h]
\centering
\includegraphics[width=0.86\linewidth]{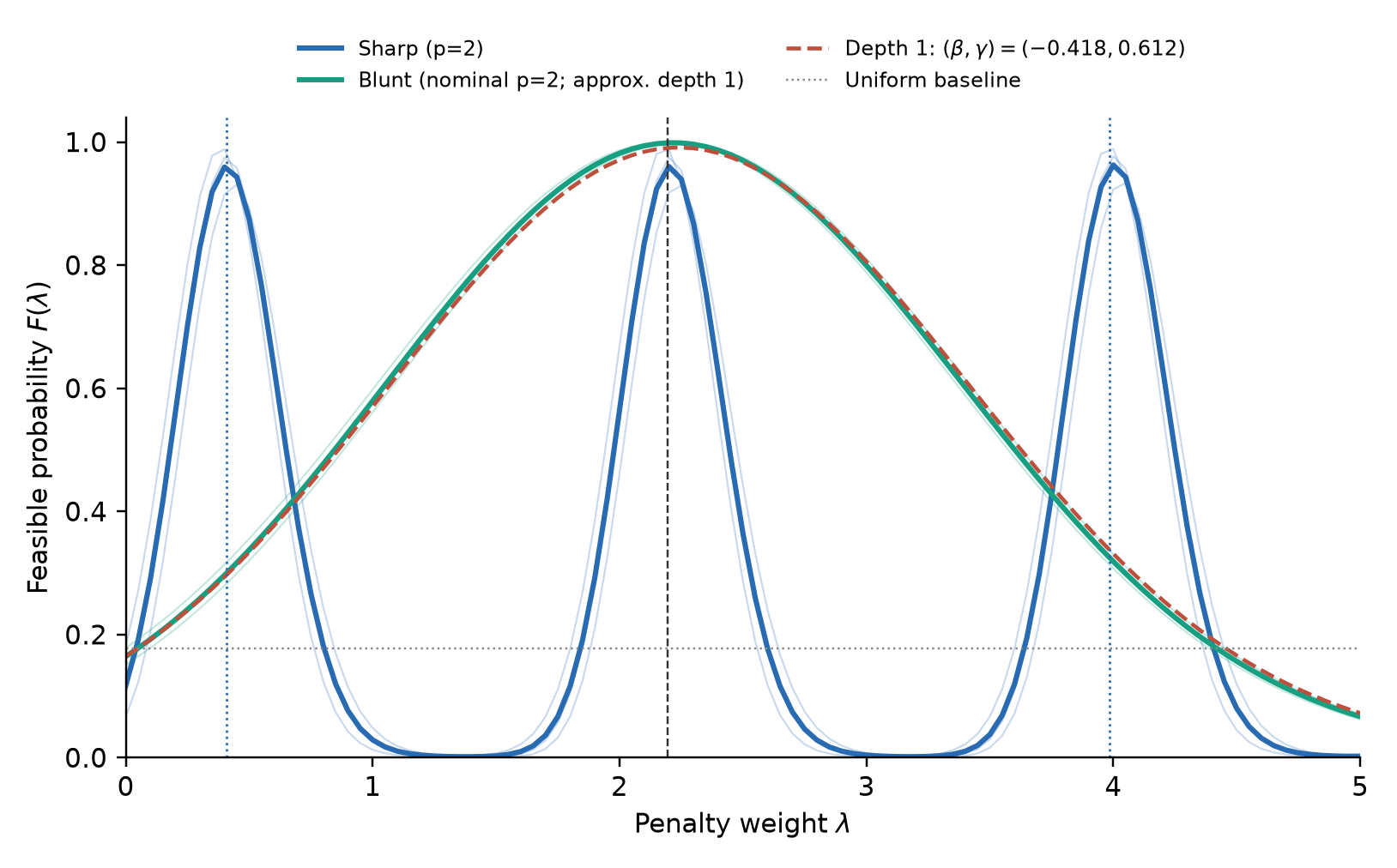}
\caption{Exact feasible fraction $F(\lambda)$ at fixed trained angles on the
20-qubit instance (thin: individual scalarization weights; bold: mean of
three).  The black dashed line marks the training penalty
$\lambda_0=2.195970387\ldots$. For the sharp set
($\sum_k|\gamma_k|=4.29$), the blue dotted lines mark the candidate
single-layer rephasing locations
$\lambda_0\pm2\pi/\gamma_1={0.406979600265\ldots}$ and
${3.984961173953\ldots}$. The blunt set
($\sum_k|\gamma_k|=0.61$) produces a broad dome.
{It is nominally $p=2$ but approximately depth one
for $F$ on this grid; the dashed control uses $p=1$ and
$(\beta,\gamma)=(-0.418,0.612)$.}
The horizontal dotted line
is the exact uniform-sampling baseline, $F=0.178$.}
\label{fig:curves}
\end{figure}

\paragraph{Q2: Empirical peak widths.}

For each angle set, let $\lambda_\ast$ denote the location of its largest
sampled value on the $[0,5]$ grid. Using the exact uniform baseline
$b=|\feas|/2^n$, we define the half-prominence level
\begin{equation}
h=b+\frac{F(\lambda_\ast)-b}{2}.
\end{equation}
Moving outward from $\lambda_\ast$ in each direction, we locate the first grid
sample at or below $h$. The separation between these two exterior samples is
denoted by $W_{\rm out}$. Since the grid spacing is $s=0.05$, the
crossing-based full width $W$ is bracketed by
\begin{equation}
W\in[W_{\rm out}-2s,W_{\rm out}].
\end{equation}
We also report $W/2$ as a summary half-width; for an asymmetric peak, this
quantity need not equal either one-sided crossing distance.
Table~\ref{tab:empirical-widths} summarizes these grid-based width measurements.
{Peak values and widths refer to the mean of the
three scalarizations. The supplementary data include a curvature-based
check validating the reported width brackets.}

\begin{table}[htbp]
\caption{Grid-based half-prominence width summaries for the sharp and blunt
angle sets. The sweep spacing is $s=0.05$.}
\label{tab:empirical-widths}
\centering
\small
\begin{tabular}{lccccc}
\toprule
set & $\lambda_\ast$ & $F(\lambda_\ast)$ & $W_{\rm out}$
& $W$ bracket & $W/2$ bracket\\
\midrule
sharp & $4.00$ & $0.962742$ & $0.45$
& $[0.35,0.45]$ & $[0.175,0.225]$\\
blunt & $2.20$ & $0.998766$ & $2.45$
& $[2.35,2.45]$ & $[1.175,1.225]$\\
\bottomrule
\end{tabular}
\end{table}

Thus, $0.45$ and $2.45$ are exterior full-width spans, not half-widths.
The corresponding half-prominence drops are $\delta=0.392382$ for the sharp
set and $\delta=0.410394$ for the blunt set. Applying
Corollary~\ref{cor:width} gives minimum one-sided detunings of $0.0305$ and
$0.2234$, respectively. Both values lie well below the corresponding
empirical half-width brackets. The measured curves are therefore consistent
with the rigorous bound, but the bound is conservative and does not predict
their widths.

The exterior spans differ by a factor of $5.44$, whereas the phase budgets
differ by a factor of $7.01$; Corollary~\ref{cor:width} does not predict that
these ratios should agree.
{The two runs also differ in optimizer budget and
$\beta$, and the blunt feasible probability is close to the depth-one
control. Their contrast therefore cannot be attributed to phase budget
alone or used to establish a general donor-ranking rule.}

\paragraph{{Q3: Exact spectrum and coefficient-based recurrence.}}
The sampled sharp curve in Fig.~\ref{fig:curves} has three resolved peak
regions in $[0,5]$:
near $0.40$, $2.20$, and $4.00$.  The outer two are close to the single-layer
rephasing values $0.4069796$ and $3.9849612$.
{
We recover all 169 torus coefficients per scalarization using
\eqref{eq:exactdft}. Denote the mean response by
$\overline F=\frac13\sum_{a=0}^{2}F_{c^{(a)}}$ and its coefficients by
$\overline A_{\mathbf m}=\frac13\sum_{a=0}^{2}A_{\mathbf m}^{(a)}$.
The exact spectrum of this mean is shown in Fig.~\ref{fig:spectrum};
Table~\ref{tab:spectral-lines} lists its five largest positive-frequency
paired amplitudes. The supplementary files give all complex coefficients
for each scalarization and their mean.

As an independent numerical check, Fourier reconstruction agrees with the
directly simulated 101-point sharp sweep for each scalarization to within
$7.8\times10^{-16}$ in absolute probability. Eight additional off-grid
phase pairs per scalarization give a maximum discrepancy below
$2.3\times10^{-16}$. These checks quantify floating-point accuracy;
the absence of aliasing follows from the finite-degree argument.

\begin{table}[htbp]
\caption{Five largest positive-frequency paired amplitudes in the exact
sharp spectrum, averaged over the three scalarizations of
\eqref{eq:sweepweights}. Each row represents a conjugate pair and therefore
reports $2|\overline A_{\mathbf m}|$, not $|\overline A_{\mathbf m}|$.
Values are rounded for display.}
\label{tab:spectral-lines}
\centering
\begin{tabular}{cccc}
\toprule
$m_1$ & $m_2$ & $\omega_{\mathbf m}=m_1\gamma_1+m_2\gamma_2$ & $2|\overline A_{\mathbf m}|$\\
\midrule
1 & 0 & 3.512140 & 0.425262\\
2 & 0 & 7.024279 & 0.192004\\
3 & 0 & 10.536419 & 0.053773\\
4 & 0 & 14.048558 & 0.009268\\
1 & 1 & 4.290272 & 0.008773\\
\bottomrule
\end{tabular}
\end{table}

\begin{figure}[h]
\centering
\includegraphics[width=\linewidth]{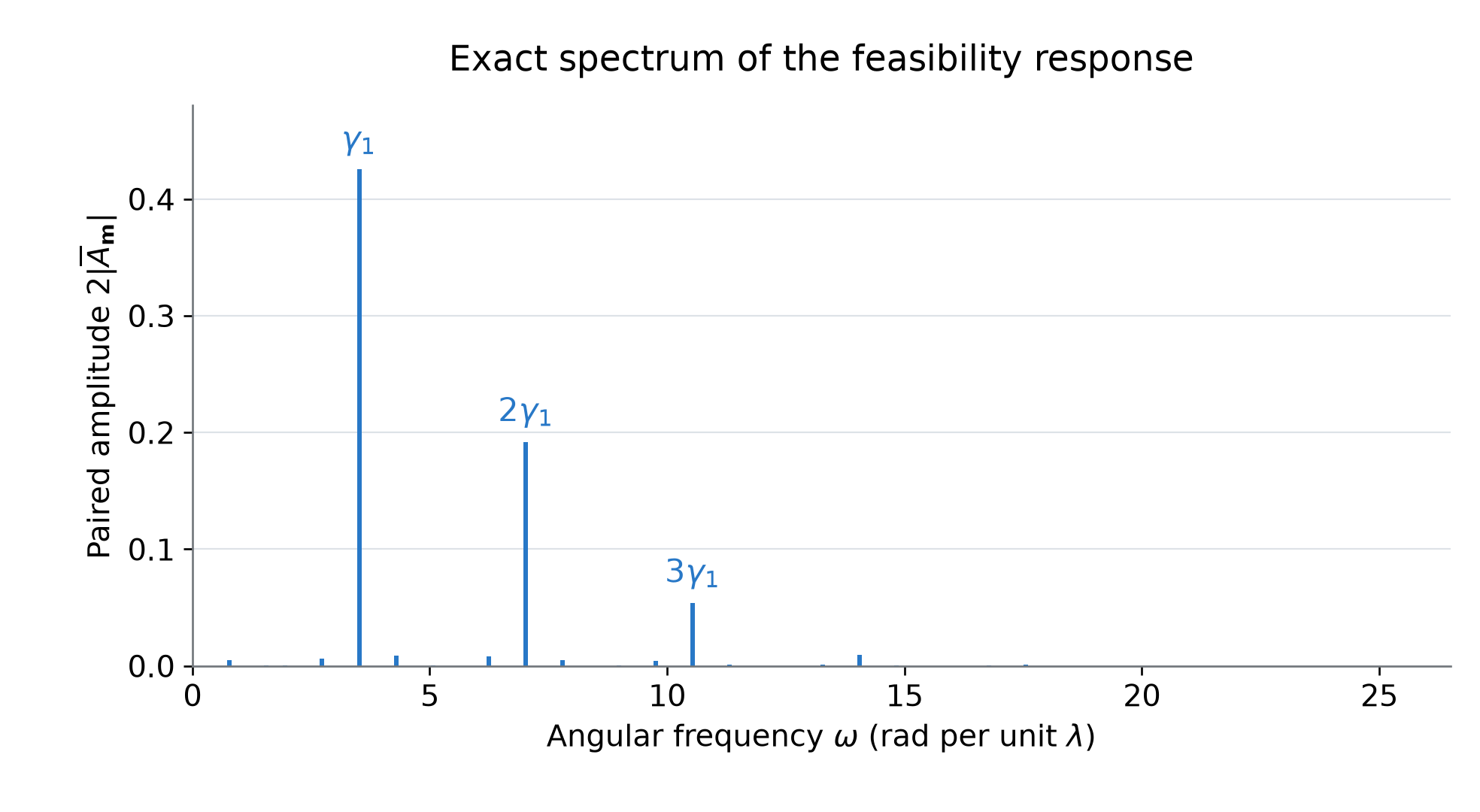}
\caption{Exact torus spectrum for the sharp angles, averaged over the three
scalarizations. The plot shows all 84 positive-frequency conjugate-pair
amplitudes $2|\overline A_{\mathbf m}|$ at
$\omega_{\mathbf m}=m_1\gamma_1+m_2\gamma_2$, $|m_1|,|m_2|\le6$;
the three largest are labelled. The constant coefficient
$\overline A_{0,0}\simeq0.280693$ is omitted.
Complex coefficients are averaged before taking magnitudes.
The coefficients are recovered without aliasing from 169 circuit
evaluations per scalarization and are numerically validated by independent
statevector evaluations.}
\label{fig:spectrum}
\end{figure}

The dominant coefficients have $m_2=0$ and rephase exactly at
$\Delta=2\pi/\gamma_1=1.788990786844\ldots$.
At this displacement, the degree-only estimate in
\eqref{eq:partialrevival} exceeds 1 before the trivial cap is applied,
but the coefficient-based form of Corollary~\ref{cor:revivals}(ii) gives
\begin{equation}
\big|\overline F(\lambda_0+\Delta)-\overline F(\lambda_0)\big|
\le \sum_{m_2\ne0}|\overline A_{\mathbf m}|
\big|e^{-i\omega_{\mathbf m}\Delta}-1\big|
<0.052275.
\label{eq:measuredrevival}
\end{equation}
Direct evaluation gives $\overline F(\lambda_0)\simeq0.959328774$ and
$\overline F(\lambda_0+\Delta)\simeq0.957742686$, an absolute difference of
$0.00158609$. The individual scalarizations yield coefficient bounds
$0.0535308$, $0.0478381$, and $0.0574656$, respectively, so each is
non-vacuous as well. These values refer to the exact rephasing location
$\lambda_0+\Delta\simeq3.984961174$.
At the nearby sampled peak $\lambda=4.00$, the $m_2=0$ terms do not rephase
exactly; the unrestricted coefficient sum over all modes is $0.105310$,
compared with a measured mean-curve difference of $0.00341361$ from
$\overline F(\lambda_0)$. Thus the recovered coefficients give quantitative
recurrence bounds without asserting that the rephasing point itself is a
local maximum.
}

\section{Cross-instance transfer probe}
\label{sec:transfer_probe}

{
We transfer parameters from 20-qubit donors to 28-qubit targets.
Here \emph{homogeneous} means that all four users use the MNIST image
modality. A \emph{mixed-modality} instance contains two MNIST-image users
and two text users, with one of each in each cell. 
The \emph{MNIST--text donor} is the 20-qubit mixed-modality training
instance together with its learned parameters.
We denote the homogeneous and mixed targets by $T_{\rm hom}$ and
$T_{\rm mix}$, respectively. The sharp and blunt parameter sets in
Table~\ref{tab:trained-angle-sets} are $D_{\rm spec,s}$ and
$D_{\rm spec,b}$, both trained at $\lambda_0=2.195970387\ldots$.
The separate mixed-modality donor $D_{\rm mix}$ was trained at
$\lambda_0=2.091719280\ldots$.
}

{For disjoint at-most-one groups, uniform sampling
has feasible probability $\prod_j(|G_j|+1)/2^{|G_j|}$.}
For ten size-$2$ groups, the exact uniform baseline is
$(3/4)^{10}=0.0563135$.

For this sweep, $\Delta\lambda=\lambda_{\rm target}-\lambda_{\rm trained}$.
A dedicated sweep over the target penalty $\lambda$ used $8$ weight vectors
and $2048$ shots per vector on an ideal simulator.
{Table~\ref{tab:target-penalty-sweep} separates this
fixed-donor, fixed-target comparison from controls that change the donor
or target composition.}

\begin{table}[htbp]
\caption{{Sparse penalty sweep and transfer controls for
the 28-qubit targets.} Each row uses eight weight vectors and 2048 shots per
vector on an ideal simulator; hv denotes pooled-front hypervolume.
{
$\dagger$ Nominally $p=2$; approximately depth one for $F$ on the
tested 20-qubit grid (Section~\ref{sec:testbed}).}}
\label{tab:target-penalty-sweep}
\centering
\small
\begin{tabular}{ccccl}
\toprule
$\Delta\lambda$ & feasible fraction & hv ratio & {donor} & {target}\\
\midrule
$-0.65$ & {$0.0002$} & 0.9456 & {$D_{\rm spec,s}$} & {$T_{\rm hom}$}\\
$-0.27$ & \textbf{0.957} & 0.9999 & {$D_{\rm spec,b}^{\dagger}$} & {$T_{\rm hom}$}\\
$-0.16$ & 0.594 & 0.9997 & {$D_{\rm mix}$} & {$T_{\rm hom}$}\\
$\phantom{-}0.00$ & {\textbf{0.965}} & {0.9995} & {$D_{\rm spec,s}$} & {$T_{\rm hom}$}\\
$+0.09$ & 0.758 & 1.0000 & {$D_{\rm spec,s}$} & {$T_{\rm mix}$}\\
$+0.31$ & 0.130 & 0.9981 & {$D_{\rm spec,s}$} & {$T_{\rm hom}$}\\
\bottomrule
\end{tabular}
\end{table}

{
For $D_{\rm spec,s}\to T_{\rm hom}$, the matched penalty gives the
highest observed feasible fraction among the three tested penalties.
The remaining three rows change the donor or target and serve as
transfer controls.
}

{At $\Delta\lambda=+0.31$, the feasible fraction is
$0.130$ while the hv ratio is $0.9981$. Thus a high pooled-front hypervolume
can coexist with a low feasible fraction in this finite-shot run. This
observation does not establish that the full conditional distribution over
feasible samples is stable.}

\subsection{\texorpdfstring{Conditional transfer protocol and empirical outcomes}{Conditional transfer protocol and empirical outcomes}}
\label{sec:protocol}


For a minimization problem, write
$C_\feas^\star=\min_{x\in\feas}C(x)$ and define the target validity threshold
\begin{equation}
\lambda_{\rm valid}=\max_{x:v(x)>0}
\left[\frac{C_\feas^\star-C(x)}{v(x)}\right]_+.
\label{eq:validgate}
\end{equation}
Any $\lambda>\lambda_{\rm valid}$ makes every infeasible state strictly more
expensive than the best feasible state. The deployment rule is therefore:

\begin{enumerate}
\item Treat $(\beta,\gamma,\lambda_0)$ and the training convention as one
  donor parameter set.  If $\lambda_0>\lambda_{\rm valid}$, use $\lambda_0$ as
  a deployment \emph{starting point} that preserves the numerical scale
  $\gamma_k\lambda_0$ under a common normalization of $v$, not as a guaranteed crest;
  bracket it with a small gate-valid sweep when circuit evaluations are
  available.
\item If $\lambda_0\le\lambda_{\rm valid}$, do not deploy the energetically
  invalid match.  Retrain at a penalty valid for both instances, or use a
  constraint-preserving mixer or encoding~\cite{hadfield2019qaoa}.
\end{enumerate}

{
For the homogeneous 28-qubit target, the empirical validity check
passed at the donor-matched penalty, comparing $4096$ infeasible
probes with a feasible training sample across $16$ validation
scalarizations. This sampled check does not certify the all-state
validity condition in \eqref{eq:validgate}.
The corresponding zero-detuning run in
Table~\ref{tab:target-penalty-sweep} yielded $15{,}814$ feasible shots
out of $16{,}384$ ($F\simeq0.965$) and a pooled-front hypervolume
ratio of $0.9995$ against the held-out reference front.
}

\section{Limitations and Future Work}
\label{sec:limitations}

The spectral representation is general under its assumptions, whereas the
empirical case study is restricted in scope, and each limitation
points directly at the next experiment or proof.

\begin{itemize}
\item \textbf{Breadth of replication.} The exact-statevector study isolates
  the response on one instance family at {nominal depth
  $p=2$, with a $p=1$ control,} and $n=20$; the 28-qubit
  evidence is shot-based and {comprises the six
  configurations in Table~\ref{tab:target-penalty-sweep}, including three
  points of a fixed-donor, fixed-target penalty comparison}. Exact tracings across
  problem families, depths, and independent training seeds are needed.
\item \textbf{Peak placement.} Energy optimization need not make
  $F'(\lambda_0)=0$, as Example~\ref{ex:flank} demonstrates.  Under the same
  normalization of $v$, matching $\lambda_0$ preserves the numerical products
  $\gamma_k\lambda_0$ but remains only a starting point for deployment validation.
\item \textbf{Width and effective depth.} The two studied angle sets differ
  in optimizer budget, all angle values, and proximity to a depth-reduction
  limit.  They cannot estimate an ensemble width relation.  A clean study
  requires common budgets and loss bands, explicit effective-depth controls,
  exact crossing searches, and uncertainty across training seeds and
  scalarizations.

\item \textbf{Cross-instance coefficients.} Transferring to a different
  instance generally changes the coefficients $A_\omega$, not just the operating point.
  The sparse transfer data are consistent with a similar response, but a
  theory of how $A_\omega$ deforms under structural transfer remains open.
\item \textbf{Non-integer penalties.} If $v$ takes values in any finite set
  $V\subset\R$, the proof goes through verbatim with
  $\Omega=\{\sum_k\gamma_k(v_k-v'_k): v_k,v'_k\in V\}$: $F$ is still a finite
  trigonometric polynomial, but {without integer-valued violations} there is no common
  coordinate period in general. This mirrors the known
  contrast between integer-valued and weighted cost Hamiltonians: generic
  weighted landscapes cease to be periodic and, under the assumptions studied
  in Ref.~\cite{sureshbabu2024parameter}, can be approximately aligned with an
  unweighted landscape by rescaling.  Quantifying approximate recurrence here
  is open.
\item \textbf{Beyond single penalties.} The path-sum argument applies to any
  parameter entering the circuit exclusively through a diagonal,
  finitely-valued operator. {Extensions to several
  penalty terms and to the weight vectors of scalarized multi-objective
  costs merit separate analysis.}
\end{itemize}

\section{Conclusion}
\label{sec:conclusion}
The exact result is a transfer-specific projection of known QAOA phase
periodicity.  When $C_\lambda=C+\lambda v$ and the trained angles are fixed,
the feasible probability is a trigonometric polynomial along the line
$\alpha=\lambda\gamma$ of the penalty-phase torus.  Unlike a uniform scaling
of the full Hamiltonian, changing only $\lambda$ cannot generally be absorbed
by rescaling $\gamma$ without also changing the objective phase.  The explicit
support and coefficient formulas make that distinction operational.

The result has deliberately limited consequences as mentioned in Sec~\ref{sec:limitations}.
{The 20-qubit phase-grid calculation recovers the realized
Fourier coefficients and makes the recurrence bound quantitative. The
depth-one control identifies the blunt curve as a near-depth-reduction
example, while the trained counterexample rules out a universal
peak-placement guarantee at the matched penalty. These results do not
establish an empirical width law or a size-scaling law.}

{The trained angles and penalty weight define a joint
phase calibration: under a common normalization of $v$, deploying at the
matched penalty preserves the training products $\gamma_k\lambda_0$.
Record $(\beta,\gamma,\lambda_0)$ together, verify target validity, and use
a valid match as the starting point for a small fixed-angle deployment
sweep. If the match is invalid, retrain at a common valid penalty or use a
constraint-preserving encoding.}
A small phase budget strengthens the rigorous worst-case tolerance
bound, but actual robustness must be measured and effective circuit depth
must be controlled.  These checks turn penalty mismatch into a testable
hypothesis without claiming that it is the only possible cause of transfer
failure.

\medskip
\noindent \textbf{Use of Generative AI: } LLMs were used to structure sentences for clarity and to suggest alternative phrasings for the manuscript. The final manuscript was carefully reviewed and corrected
by the authors to ensure precision.

\medskip
\noindent \textbf{Acknowledgements:} This work was funded by the Natural Sciences and Engineering Research Council of Canada (NSERC)
through an Undergraduate Student Research Award (USRA). This work was also supported by the ChatGPT for Academic Researchers Program.

\medskip
\noindent \textbf{Data availability}
{
The data and code supporting the reported results are included in
the supplementary materials,
\texttt{lambda\_resonance\_reproducibility.zip}. Code and execution instructions
are provided for the depth-one comparison, exact spectrum,
recurrence bounds, and generation of the figures and tables from
the supplied data. Numerical validation results are also included.
}

\end{document}